\documentclass[aps,pra,twocolumn]{revtex4-2}
\usepackage{amsmath,amsfonts,amssymb,amsthm}
\usepackage{mathtools}
\usepackage[pdftex]{graphicx}
\usepackage[usenames, dvipsnames, svgnames, table]{xcolor}
\usepackage[colorlinks=true, citecolor=Blue, linkcolor=BrickRed, urlcolor=Brown]{hyperref}
\usepackage{txfonts}
\usepackage{mathrsfs}
\usepackage{bm}
\usepackage{multirow}
\usepackage{braket}
\usepackage{import}
\usepackage{qcircuit}
\usepackage[ruled, noline, noend]{algorithm2e}
\usepackage{array}
\usepackage{comment}
\usepackage{tikz}

\newtheorem{theorem}{Theorem}
\newtheorem{lemma}{Lemma}

\newtheorem{Note}{Note}

\theoremstyle{definition}
\newtheorem{definition}{Def.}
\newcommand{\be}{\begin{equation}}
\newcommand{\ee}{\end{equation}}
\newcommand{\ben}{\begin{eqnarray}}
\newcommand{\een}{\end{eqnarray}}
\newcommand{\bes}{\begin{subequations}}
\newcommand{\ees}{\end{subequations}}
\newcommand{\bF}{\begin{figure}}
\newcommand{\eF}{\end{figure}}

\DeclareMathAlphabet{\pazocal}{OMS}{zplm}{m}{n}

\newcommand{\Q}{\pazocal{Q}}

\newcommand{\orcid}[1]{\href{https://orcid.org/#1}{\includegraphics[height = 2ex]{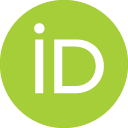}}}

\begin{document}

\title{Explaining the Ubiquity of Phase Transitions in Decision Problems}

\author{Andrew Jackson \orcid{0000-0002-5981-1604}}
\affiliation{Department of Physics, University of Warwick, Coventry CV4 7AL, United Kingdom}
\date{\today}


\begin{abstract}
     I present an analytic approach to establishing the presence of phase transitions in a large set of decision problems. This approach does not require extensive computational study of the problems considered. The set -- that of all paddable problems over even-sized alphabets satisfying a condition similar to not being sparse -- shown to exhibit phase transitions contains many ``practical" decision problems, is very large, and also contains extremely intractable problems.
\end{abstract}

\maketitle
\section{Introduction}
\label{IntroSec}
Phase transitions are the epitome of emergent behavior in both many-body physics and theoretical computer science.
A phase transition in a decision problem -- as in a physical system -- is a sudden/rapid change in the behavior of the problem (or physical system) as a specific parameter changes and this typically -- in decision problems, that may be considered WLOG as deciding if a string is in a given language -- takes the form of a sudden change in the probability a random input, with a specific value of the parameter, is in the relevant language. They are of great interest~\cite{saitta_giordana_cornuejols_antoine_2011} to computer scientists and mathematicians with ``a crucial role in the field of artificial intelligence and computational complexity theory"~\cite{FLAXMAN2004301}.

Phase transitions are also vitally important in the physics of condensed matter/many-body systems --  and consequently are of interest to physicists~\cite{domb2000phase, bruce1981structural, papon2002physics} -- but are an emergent property of many-body systems and hence are not easily studied theoretically. 
Their importance, coupled with the complexity of deriving their existence and properties, results in emerging computational techniques and increasing computational power continually being applied~\cite{PhysRevB.94.195105} to their study~\footnote{E.g. the use of quantum computers to evaluate partition functions at complex temperatures, looking for their zeros -- and hence phase transitions~\cite{PhysRevA.107.012421}}. It is therefore hoped that investigations of phase transitions in decision problems could shed light on their condensed matter cousins. As phase transitions also appear even further afield, such as in biology~\cite[Chapter~6]{stein_newman_2013} and sociology~\cite{doi:10.1080/0022250X.1982.9989929}, investigations of phase transitions may find applicability in many fields not conventionally associated with phase transitions.

Given the paramount significance of phase transitions in such a wide array of fields, the formidable computational barriers to investigations based on ``experimental" research (i.e. by the computation of a large number of instances of a problem and looking for a phase transition e.g. \cite{1inKSAT, Num.Ev, RandomGraph, Schawe_2016}), which provide the foundation of our current understanding of phase transitions, may be severely limiting to these fields.
These investigations tend to be computationally intensive, requiring extensive sampling of the problem. It is then not hard to see that theoretical, non-computational, approaches to identifying problems with phase transitions -- based on easily checked features of the problems -- would be useful.
This is especially true for problems with complexity substantially greater than polynomial time and/or with very few instances much easier than the worst case.

However, these pre-existing computational studies do provide a good background on phase transitions. We know, for instance, that phase transitions are not too rare, being observed in many decision problems.
Most notably, NP-complete decision problems often have phase transitions, but phase transitions have been observed in much harder problems~\cite{Gent1994TheSP}.
For example, problems observed to have phase transitions include: playing Minesweeper~\cite{minesweeperPhase}, K-colourability~\cite{https://doi.org/10.1002/(SICI)1098-2418(1999010)14:1<63::AID-RSA3>3.0.CO;2-7}, 3 and 4 - SAT~\cite{Gent1994TheSP}, protein folding \cite[Chapter~6]{stein_newman_2013}, and the dynamics of strike action~\cite{doi:10.1080/0022250X.1982.9989929}.

In this paper, I start, in Sec.~\ref{Def.sPrelimSubsection}, by setting out my notation and reviewing the basic concepts of decision problems and their phase transitions. Sec.~\ref{defAndPrelim} then builds towards introducing a formal definition (Def.~\ref{PhaseDef}) that captures a subset of phase transitions. I then proceed, in Sec.~\ref{MainResSubsec} via the work of Faragó~\cite{farago2016roughly}, to show (in Theorem~\ref{mainResultTheorem}) that phase transitions occur in a large number of decision problems (those which are not-anywhere-exponentially-unbalanced, paddable~\footnote{See Def.~\ref{padDef}.}, and over even-sized alphabets).
\section{Preparation and Preliminary Results}
\subsection{Preliminary Definitions}
\label{Def.sPrelimSubsection}
\begin{definition}
    An \underline{alphabet} is a finite set of distinct symbols e.g. $\{a, b, c, d, e, f, g\}$. Throughout this paper, alphabets will generally be denoted by $\Sigma$ and I will assume $\vert \Sigma \vert$ (the size of the alphabet $\Sigma$) is both even and at least two.
\end{definition}
\begin{definition}
    Let $\Sigma$ be any alphabet, then I define \underline{$\Sigma^*$} as the set of all strings -- including the empty string -- of symbols from $\Sigma$. 
\end{definition}

\begin{definition}
    For any alphabet, $\Sigma$, $x$ is a \underline{word over $\Sigma$} $\iff$ $x \in \Sigma^*$.
\end{definition}

\begin{definition}
    $\mathcal{L}$ is a \underline{language over $\Sigma$} $\iff$ $\mathcal{L} \in \text{POW}\big( \Sigma^* \big)$,
    where $\text{POW}\big( \Sigma^* \big)$ denotes the power set of $\Sigma^*$.
    I.e. a language is a set of words over $\Sigma$ .
\end{definition}

\begin{definition}
    A \underline{decision problem} is any problem that gives a binary (e.g. yes/no) answer.
    For the purposes of this paper, decision problems will exclusively be restricted to deciding if a given word is in a given language.
\end{definition}

\begin{definition}
    An algorithm, $\mathbb{A}$, that solves a problem, $\mathbb{P}$, has \underline{complexity}, $\mathcal{O}(f[N])$, if for every instance, $p \in \mathbb{P}$, of size $N$, $\mathbb{A}$ takes $\mathcal{O}( f[N] )$ time. 
\end{definition}

\begin{definition}
Let $\mathbb{P}$ be a decision problem and A be the set of all algorithms that solve $\mathbb{P}$. Let  $\mathbb{A}$ be the algorithm in A with the smallest (in the asymptotic limit~\footnote{Meaning in the limit as the instance size tends to infinity.}) complexity, then the \underline{complexity of $\mathbb{P}$} is the complexity of $\mathbb{A}$.
\end{definition}

\begin{definition}
    A function is said to be \underline{polynomial time computable} if there exists some polynomial function, $f_{\textit{poly}}$, such that for any possible input to the computation, $x$, the run-time of the computation is upper bounded by $f_{\textit{poly}} \big( \vert x \vert \big)$, where $\vert x \vert$ is the length of $x$.
\end{definition}

I now switch focus and present the definitions required for a discussion of phase transitions in decision problems.
\begin{definition}
\label{OGA}
   $ \forall \mathcal{S} \subseteq \Sigma^*$, I define the \underline{accepting fraction}, $\mathcal{A}[\mathcal{S}]$, relative to some specific language, $\mathcal{L}$, to be the fraction of $\mathcal{S}$ that is in $\mathcal{L}$.
\end{definition}

\begin{definition}
    \label{orderParameterDef}
    Any polynomial time computable mapping from $\Sigma^*$ to $\mathbb{R}$ is a \underline{parameter}. 
\end{definition}

\begin{definition}
    \label{def:parameterSlice}
    Let  $\gamma (x): \Sigma^* \longrightarrow \mathbb{R}$ be a polynomial time computable parameter. Define the \underline{parameter slice}, $\mathcal{S}^{\gamma}_{n} \subseteq \Sigma^*$, as the set of all $x \in \Sigma^*$ such that the parameter, $\gamma$, takes the value $n \in \mathbb{R}$.
\end{definition}

\subsection{Phase Transition Definitions and Preliminary Results}
\label{defAndPrelim}
I am now ready to start working towards a definition of phase transitions. The aim is to capture the behavior seen in computational experiments, i.e. phenomena as in Fig.~\ref{fig:Example Image}, but I would prefer to also capture a more broad range of phenomena that I would also consider to be phase transitions -- so that I am not presuming exactly what form a phase transition must take. Hence, I use Fig.~\ref{fig:Example Image} as a template but broaden the scope of this investigation slightly.
\begin{figure}[h!]
    \centering
\includegraphics[width=0.4\textwidth]{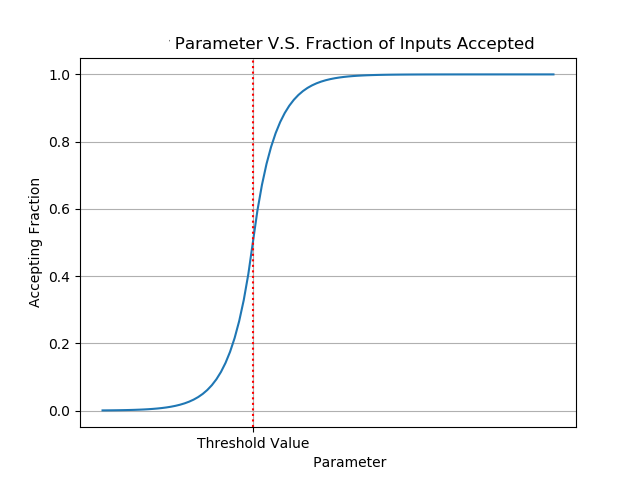}
    \caption{An Example of a Canonical Form Phase Transition}
    \label{fig:Example Image}
\end{figure}
Images similar to Fig.~\ref{fig:Example Image} can be seen in experimental data and examples are found in:\\ $\bullet$ Ref. \cite{minesweeperPhase} (Figure 4),\\ 
$\bullet$ Ref. \cite{Gent1994TheSP} (Figures 2, 4, 6, 7),\\ $\bullet$ Ref. \cite{101371journalpone0215309} (Figures 2, 3, 4),\\
$\bullet$ Ref. \cite{cross3Sat} (Figures 1, 3),\\ 
$\bullet$ Ref. \cite{BAILEY20071627} (Figures 1, 3).\\
Note these are not the only cited papers to have phase transitions of this form (and this is far from a complete list of all problems with phase transitions of this form), but they have the transition plotted so as to easily see it is of the same/similar form as Fig.~\ref{fig:Example Image}.
Other examples, less clearly of the form required, are in:\\
$\bullet$ Ref. \cite{saitta_giordana_cornuejols_antoine_2011} (Figures 9.7, 11.11, B.1),\\
$\bullet$ Ref. \cite{stanley_1971} (Figures 2.6(d), 2.8(c), 3.3(b),  6.2),\\
I attempt to consolidate the observations of phase transitions into a precise definition -- in Def.~\ref{PhaseDef} -- but first I must define a notion of ``where'' a phase transition occurs, in Def.~\ref{thresholdDef}.
\begin{definition}
\label{thresholdDef}
    A \underline{threshold value}, $\mathcal{T}$, is the closest I can come to assigning a value of a given parameter as ``where'' a phase transition occurs. Define a threshold value as the minimum $\mathcal{T} \in \mathbb{R}$ such that:
    $\forall y \in \mathbb{R},$
    \begin{align}
        n > \mathcal{T} \iff \mathcal{A}[\mathcal{S}^{\gamma}_{n}] > 0.5.
    \end{align}
\end{definition}
\begin{definition}
\label{PhaseDef}
    A language, $\mathcal{L} \subseteq \Sigma^*$, exhibits a \underline{phase transition} if and only if there exists a parameter, $\gamma: \Sigma^* \rightarrow \mathbb{R}$, such that:
    \begin{enumerate}
        \item As $n \rightarrow \infty$, $\mathcal{A}[\mathcal{S}^{\gamma}_n]\rightarrow 1$, with a monotonically increasing lower bound.
    \item As $n \rightarrow -\infty$, $\mathcal{A}[\mathcal{S}^{\gamma}_n]\rightarrow 0$, with a monotonically decreasing upper bound.   
    \item  The fraction of input strings that take values of $\gamma$ between A and A + $\delta$ -- for a small fixed $\delta$ -- grows at least exponentially (or at least does beyond a short distance from the threshold value) as $\vert$ A - $\mathcal{T}$ $\vert$ increases, where $\mathcal{T} \in \mathbb{R}$ is the threshold value.
    \end{enumerate}
\end{definition}
I note that the more general definition of a phase transition is not normally as prescriptive about the third of the above requirements, but I note Def.~$\ref{PhaseDef}$ meets the standard definition which usually just requires that:\\
 
``Few input strings have values of the parameter close to the threshold value (shown in Fig.~\ref{fig:Example Image} and defined in Def.~\ref{thresholdDef})."\\

It is important to stress that the phase transition in Fig.~\ref{fig:Example Image} is just a typical example. Def.~\ref{PhaseDef}, which is use throughout this paper, admits more general phase transitions. For example, while Fig.~\ref{fig:Example Image} is continuous, Def.~\ref{PhaseDef} does not require continuity.\\

Thus far I have assumed an orientation of the phase transition. Namely, that:\\
 $\bullet$ As $n \rightarrow \infty$, $\mathcal{A}[\mathcal{S}^{\gamma}_n]\rightarrow 1$.\\
 $\bullet$ As $n \rightarrow -\infty$, $\mathcal{A}[\mathcal{S}^{\gamma}_n]\rightarrow 0$.\\
Rather than:\\
$\bullet$ As $n \rightarrow \infty$, $\mathcal{A}[\mathcal{S}^{\gamma}_n]\rightarrow 0$.\\
 $\bullet$ As $n \rightarrow -\infty$, $\mathcal{A}[\mathcal{S}^{\gamma}_n]\rightarrow 1$.\\
Which would equally be a valid phase transition. I refer to the orientation assumed in Def.~\ref{PhaseDef} as the canonical form and the newly mentioned one as the inverted form. Def.~\ref{PhaseDefInverted} -- which defines an inverted form phase transition -- can be seen to be almost identical to the definition of a canonicalform phase transition in Def.\ref{PhaseDef}.
\begin{definition}
\label{PhaseDefInverted}
    A language, $\mathcal{L}$, is said to exhibit an \underline{inverted form phase transition} if and only if there exists a polynomial time calculable parameter, $\gamma: \Sigma^* \longrightarrow \mathbb{R}$, such that:
    \begin{enumerate}
    \item  As $n \rightarrow \infty$, $\mathcal{A}[\mathcal{S}^{\gamma}_n]\rightarrow 0$, with a monotonically decreasing bound.
    \item  As $n \rightarrow -\infty$, $\mathcal{A}[\mathcal{S}^{\gamma}_n]\rightarrow 1$, with a monotonically increasing bound.
    \item  The fraction of input strings that take values of $\gamma$ between A and A + $\delta$ grows at least exponentially (or at least does beyond a short distance from the threshold value) (in  $\vert$ A - $\mathcal{T}$ $\vert$) as $\vert$ A - $\mathcal{T}$ $\vert$ increases, where $\mathcal{T} \in \mathbb{R}$ is the threshold value -- albeit adapted to the inverted form -- and $\delta$ is fixed.
    \end{enumerate}
\end{definition}

However, I do not -- from here on -- consider inverted form phase transitions due to Lemma $\ref{inverted}$.
\begin{lemma}
\label{inverted}
For every inverted form phase transition there exists a canonical form phase transition describing the exact same phenomena -- and vice versa.
\end{lemma}
\begin{proof}
    Assume that there exists the required inverted form phase transition, with all notation as defined in Def.~\ref{PhaseDefInverted}.
    I begin with the asymptotic behavior requirements.
    As per Def.~\ref{PhaseDefInverted}, an inverted form phase transition has the asymptotic limits:
    \begin{align}
        \lim_{\tau \longrightarrow -\infty} \bigg( \mathcal{A} \big[ \mathcal{S}^{\Gamma}_\tau \big] \bigg) = 1 \text{ and }
        \lim_{\tau \longrightarrow \infty} \bigg( \mathcal{A} \big[ \mathcal{S}^{\Gamma}_\tau \big] \bigg) = 0,
    \end{align}
    where $\Gamma: \Sigma^* \rightarrow \mathbb{R}$ is the parameter that the inverted phase transition exists relative to. Define a parameter, $\Gamma^{\prime} : \Sigma^* \rightarrow \mathbb{R}$, by:
    \begin{align}
        \label{eqn:DefiningInvertedParameter}
        \Gamma^{\prime}(x) &= - \Gamma(x).
    \end{align}
    Using this new parameter, in Eqn.~\ref{eqn:DefiningInvertedParameter}, and letting $\mathcal{A}''$ be the acceptance fraction according to it (i.e. $\forall \tau' \in \mathbb{R}$, $\mathcal{A}''[\mathcal{S}^{\Gamma}_{\tau'}] = \mathcal{A}[\mathcal{S}^{\Gamma}_{- \tau'}]$), the asymptotic limits are:
    \begin{align}
        \lim_{\tau^{\prime} \longrightarrow \infty}\bigg( \mathcal{A}'' \big[ \mathcal{S}^{\Gamma'}_{\tau^{\prime}} \big] \bigg)
        =
        1 \text{ and }
        \lim_{\tau^{\prime} \longrightarrow -\infty}\bigg( \mathcal{A}'' \big[ \mathcal{S}^{\Gamma'}_{\tau^{\prime}} \big] \bigg)
        =
        0.
    \end{align}
    Hence the asymptotics are now as required for a canonical form phase transition. The only remaining thing to consider is that the sparsity conditions (around the threshold value) required of either form of phase transition is unaffected by this redefinition of the parameter: if the density increases exponentially in both directions then flipping the two directions leaves the density still increasing exponentially in both directions.
    
    Therefore, all requirements of a canonical form phase transition are forfilled.
    That for every canonicalform phase transition there exists an inverted form phase transition follows from a similar argument.
\end{proof}
\subsection{Main Result}
\label{MainResSubsec}
\subsubsection{Definitions: Paddability and RoughP}
I first present a review of the notation, concepts, and some results of Ref.~\cite{farago2016roughly}, beginning with the required notions of a P-isomorphism, in Def.~\ref{def:PIsomorph}, and its P-isomorphism output size, in Def.~\ref{pIsoDef}.
\begin{definition}
    \label{def:PIsomorph}
    For any alphabet, $\Sigma$, a bijection from and to $\Sigma^*$ that can be both computed and inverted in polynomial time is referred to as a \underline{P-isomorphism}.   
\end{definition}
\begin{definition}
    \label{pIsoDef}
    Given a P-isomorphism, $\xi: \Sigma^* \rightarrow \Sigma^*$, define the \underline{P-isomorphism output size} of $x\in \Sigma^*$, $N_{\xi}$, as the length of the string $\xi(x)$. i.e. $\forall x \in \Sigma^*$,
    \begin{align}
        N_{\xi}(x)
        &=
        \big \vert \xi(x) \big \vert.
    \end{align}
\end{definition}
The final key ingredient of the main result of Ref.~\cite{farago2016roughly}, that I will need for this paper is paddability, presented in Def.~\ref{padDef}.
\begin{definition}
    \label{padDef}
    A language, $\mathcal{L} \subseteq \Sigma^*$, is \underline{paddable} if and only if there exists two polynomial time computable functions:
\begin{enumerate}
        \item $\textit{Pad}: \Sigma^* \times \Sigma^* \longrightarrow \Sigma^*$,
        \item $\textit{Dec}:  \Sigma^* \longrightarrow \Sigma^*$,
\end{enumerate}
    such that, $\forall x, y \in \Sigma^*$:
    \begin{enumerate}
        \item $\textit{Pad}(x, y) \in \mathcal{L} \iff x \in \mathcal{L}$,
        \item $\textit{Dec}(\textit{Pad}(x, y)) = y$.
    \end{enumerate}
\end{definition}
I am now ready to present the main result of Ref.~\cite{farago2016roughly}, which will be vital for the results of the current paper.
\begin{theorem}{(Theorem 1 in Ref.~\cite{farago2016roughly})}
\label{faragoRoughMain}
    For any paddable language, $\mathcal{L}$, there exists a polynomial time algorithm, called a RoughP algorithm, $\mathcal{P}: \Sigma^* \longrightarrow \{ \textit{Accept }, \textit{ Reject}, \perp \}$, that either correctly decides $\mathcal{L}$ or returns a ``do not know" symbol, $\perp$. Additionally, there exists a P-isomorphism, $\phi: \Sigma^* \rightarrow \Sigma^*$, such that:
        $\forall n \in \mathbb{N}, \forall x \in \Sigma^*$, the fraction of $\big \{ x \in \Sigma^* \text{ } \vert \text{ } \phi(x) \in \Sigma^n \big \}$ such that $\mathcal{P}(x) = \perp$ is at most $2^{-n/2}$.
\end{theorem}
Throughout the rest of this paper $\phi: \Sigma^* \rightarrow \Sigma^*$ will continue to represent the P-isomorphism it does in Theorem~\ref{faragoRoughMain}, as it does in Def.~\ref{def:RoughP}.
\begin{definition}{(Def.~3 in Ref.~\cite{farago2016roughly})}
    \label{def:RoughP}
    Let,\\
    $\bullet$ $\Sigma$ be an alphabet with $\vert \Sigma \vert \geq 2$,\\
    $\bullet$ $\mathcal{L} \subseteq \Sigma^*$
be a language.\\
I say that $\mathcal{L} \in$ \underline{RoughP}, if there exists a P-isomorphic encoding, $\phi$, and a polynomial time algorithm,
$\mathcal{P}$ : $\Sigma^* \rightarrow$ \{Accept, Reject, $\perp$\}, such that the following hold:
\begin{enumerate}
    \item $\mathcal{P}$ correctly decides $\mathcal{L}$, as an errorless heuristic. That is, it never outputs a wrong decision:
if $\mathcal{P}$ accepts a string $x \in \Sigma^*$, then $x$  $\in \mathcal{L}$ always holds, and if $\mathcal{P}$ rejects $x \in \Sigma^*$, then $x \not \in \mathcal{L}$ always holds.
\item Besides Accept/Reject, $\mathcal{P}$ may output the symbol, $\perp$, meaning it is unable to decide if the input is in the language.
This can occur, however, only for at most an exponentially small fraction of strings. I.e. there is a constant $c$ (where $0 \leq c < 1$) such that: $\forall n \in \mathbb{N}$,
\begin{align}
    \label{eqn:RoughPDefiningEquation}
    \dfrac{\vert \mathcal{B}_n^{\phi} \cap \{ x \in \Sigma^* \text{ }\vert\text{ } \mathcal{P}(x) = \perp\} \vert}{\vert \mathcal{B}_n^{\phi} \vert} \leq c^n,
\end{align}
where \underline{$\mathcal{B}_n^{\phi}$} is defined as:
    $\mathcal{B}_n^{\phi}~=~\big \{ x \in \Sigma^* \text{ } \vert \text{ } N_{\phi}(x) = n \big\}~=~ \mathcal{S}^{N_{\phi}}_n$.
I.e. it is the pre-image of the ``ball" of radius $n \in \mathbb{N}$ in the image of $\phi$.
\end{enumerate}
\end{definition}
I also introduce another definition, that uses the same notation as immediately above.
\begin{definition}
    \label{PhiBalanced}
    A language is \underline{not-anywhere-exponentially-unbalanced} if there exists some polynomial, $\textit{Poly}: \mathbb{N} \rightarrow \mathbb{R}$, such that $\forall n \in \mathbb{N}$, neither the fraction of $\mathcal{B}_n^{\phi}$ (which is exactly the same as in Def.~\ref{def:RoughP}) that is in the language (and $\mathcal{P}$ decides correctly) nor the fraction not in the language (and $\mathcal{P}$ decides correctly) are less than $\big( \textit{Poly}(n) \big)^{-1}$, with the additional condition that $\textit{Poly}(n) \cdot \big( 1 / \sqrt{2} \big)^n$ is monotonically decreasing.
\end{definition}
Not-anywhere-exponentially-unbalanced languages are most simply viewed -- for the purposes of this paper -- as a subset of paddable languages. Paddable languages are extremely common and I do not expect that not-anywhere-exponentially-unbalanced are rare within the paddable languages.

\subsubsection{Main Result}

As mentioned in Theorem~\ref{faragoRoughMain}, Ref.~\cite{farago2016roughly} proved that all paddable languages have a RoughP algorithm, $\mathcal{P}(x)$, and from this RoughP algorithm I define the discriminator as in Def.~\ref{discrimDef}, which exists for all paddable languages.
\begin{definition}
\label{discrimDef}
The \underline{discriminator}, $\Q: \Sigma^* \rightarrow \{ +1, -1 \}$, for a given paddable language with a RoughP algorithm, $\mathcal{P}: \Sigma^* \rightarrow \{ \textit{Accept }, \textit{ Reject}, \perp \}$, is defined by:
    \begin{align}
        \Q(x)=
\begin{cases}
            +1, & \textit{if } \mathcal{P}(x) = \textit{Acc}\\
    \mathfrak{Q}^{\prime}(\phi(x)), & \textit{if } \mathcal{P}(x) = \perp\\
-1, & \textit{if } \mathcal{P}(x) = \textit{Rej}
\end{cases},
\end{align}
where $\phi$ is as in Theorem~\ref{faragoRoughMain} and $\mathfrak{Q}^{\prime}$ is as in Lemma~\ref{devidingLemma} (in Appendix~\ref{app:ProofOfLemma}) but equally divides the $x \in \Sigma^*$ such that $\mathcal{P}(x) = \perp$ equally between returning $+1$ and $-1$ (as shown in Lemma~\ref{devidingLemma}). As $\mathcal{P}$ and $\phi$ are specific to a given problem/language so is $\Q$.
\end{definition}
I am then ready for the main result of this paper, which is Theorem~\ref{mainResultTheorem}.
\begin{theorem}
\label{mainResultTheorem}
    Any paddable not-anywhere-exponentially-unbalanced language over an even-sized alphabet exhibits a phase transition.
\end{theorem}
\begin{proof} 
    To show a phase transition exists, I first must define the parameter I aim to prove induces a phase transition. The particular parameter I use will be referred to as the canonical parameter, introduced in Def.~\ref{def:canonParam}.
    \begin{definition}
        \label{def:canonParam}
        The \underline{canonical parameter}, $\Gamma: \Sigma^* \rightarrow \mathbb{R}$, for an input, $x \in \Sigma^*$, is defined (relative to a given paddable not-anywhere-exponentially-unbalanced language over an even-sized alphabet) as:
        \begin{align}
            \label{tauDef}
            \Gamma(x) = \Q(x) \bigg \vert \sqrt{N_{\phi}(x)} \bigg \vert,
        \end{align}
        where,\\
        $\bullet$ $\Q: \Sigma^* \rightarrow \{ 1, -1 \}$ is the discriminator for the relevant problem (as in Def.~\ref{discrimDef}).\\
        $\bullet$ $\phi: \Sigma^* \rightarrow \Sigma^*$ is the P-isomorphism shown to exist in Theorem~\ref{faragoRoughMain}, for the specific paddable language being considered.\\
        $\bullet$ $N_{\phi}: \Sigma^* \rightarrow \mathbb{N}$ is 
        the P-isomorphism output size of $\phi$ (as in Def.~\ref{pIsoDef}).\\
        Finally, I add the proviso that if $N_{\phi} = 0$, the canonical parameter is not defined~\footnote{This corresponds the the single case where otherwise the canonical parameter is exactly at the threshold value and I omit the study of this single input for simplicity. This does not affect the result more generally, as can be seen in Fig.\ref{fig:ExampleImageProvenSoFar} where no bound at all is placed on the acceptance at the threshold value.}. Hence the single element of $\Sigma^*$ where $\Gamma$ (there is only one as $\phi$ is an isomorphism) is not defined is neglected. This is of little consequence.
    \end{definition}
   Def.~\ref{def:canonParam} requires a discriminator $\Q$ (as in Def.~\ref{discrimDef}) to always exist for paddable problems, which follows from Ref.~\cite[Theorem 1]{farago2016roughly} (and hence also from Theorem~\ref{faragoRoughMain} herein) as the $\mathcal{P}$, corresponding to the required $\Q$, is the RoughP algorithm described therein -- that always exists for paddable languages -- and that $\mathfrak{Q}'$ always exists (and is efficient) for even-sized alphabets.
   
    I note from the above definition of $\Q$, in Def.~\ref{discrimDef}:
    \begin{align}
    \bullet\Q(x) = +1 \text{ if and only if } \Gamma(x) > 0.\\
    \bullet \Q(x) = -1 \text{ if and only if } \Gamma(x) < 0.
     \end{align}
    Which then implies:
    \begin{align}
        \Q(x) = +1 
        &\iff
        \Gamma(x) \geq 0
        \iff H[\Gamma(x)] = 1,
    \end{align}
    where $H: \mathbb{R} \rightarrow \{0, 1\}$ is the Heaviside step function~\cite{Hside}. Using the definition of RoughP, Ref.~\cite{farago2016roughly}, and the definition of $\Q$ (see Def.~\ref{discrimDef}); Lemma~\ref{AccBoundLemma} -- with the aid of Note~\ref{note:QvaluesInAppendix} -- in Appendix \ref{app:ProofOfLemma} shows~\footnote{Relying on $\Sigma$ being even-sized} that:
    \begin{align}
    &\bullet \text{If } \Q(x) = +1, \text{ then } \mathcal{A}[\mathcal{S}_{\tau}^{\Gamma}] \geq 1 - \textit{Poly}(N_{\phi}) c^{N_{\phi}}.\\
    &\bullet \text{If } \Q(x) = -1,  \text{ then } \mathcal{A}[\mathcal{S}_{\tau}^{\Gamma}] \leq \textit{Poly}(N_{\phi}) c^{N_{\phi}}.
    \end{align}
    Where $\textit{Poly}(N_{\phi})$ is a polynomial function (as in Def.~\ref{PhiBalanced}), $N_{\phi}$ denotes the length of any $\phi(x)$ when $ \big \vert \Gamma (x) \big \vert$ (as defined in Def.~\ref{tauDef}) takes the value $\tau \in \mathbb{R}$ (which is consistent due to Eqn.~\ref{nDefinTau}), and $c \in \mathbb{R}^+$ is a constant less than $1 / \sqrt{2}$.
    It is now useful to define, in Def.~\ref{def:acceptanceBounding}, a function that I will spend the rest of this paper examining.
    \begin{definition}
        \label{def:acceptanceBounding}
        An \underline{acceptance bounding function},\\ $\mathcal{A}': \text{POW}\big( \Sigma^* \big) \rightarrow [0,1]$, is defined by~\footnote{Where $\text{POW}$ denotes the powerset of its argument.}:
        \begin{align}
            \mathcal{A}'[\mathcal{S}^{\Gamma}_{\tau}] &> \mathcal{A}[\mathcal{S}^{\Gamma}_{\tau}] \hspace{0.25 cm} \text{ if } \tau > 0,\\
            \mathcal{A}'[\mathcal{S}^{\Gamma}_{\tau}] &< \mathcal{A}[\mathcal{S}^{\Gamma}_{\tau}] \hspace{0.25 cm} \text{ if } \tau < 0,
        \end{align}
        where $\mathcal{A}$ is still as in Def.~\ref{OGA}.
    \end{definition}
    Due to Lemma~\ref{AccBoundLemma} it is useful, for my purposes, to define the acceptance bounding function, $\mathcal{A}'[\mathcal{S}^{\Gamma}_{\tau}]$, as:
    \begin{align}
        \label{eqn:A'Cases}
        \mathcal{A}'[\mathcal{S}^{\Gamma}_{\tau}]
        &=
        \begin{cases}
        1 - \textit{Poly}(N_{\phi}) c^{N_{\phi}}, &\textit{if } \tau > 0\\
          \textit{Poly}(N_{\phi}) c^{N_{\phi}}, &\textit{if } \tau < 0
        \end{cases},
    \end{align}
    where $\textit{Poly}(N_{\phi})$ is the same polynomial function as in Def.~\ref{PhiBalanced}.
    
    Note that the P-isomorphism output size, $N_{\phi}(x)$, is uniquely determined by $\tau$, due to Eqn.~\ref{def:canonParam}, therefore from here on I simply write $N_{\phi}$ to denote $N_{\phi}(x)$ where $x \in \Sigma^*$ such that $\big \vert \Gamma(x) \big \vert = +\tau$, without reference to a specific $x \in \Sigma^*$.
    
    Eqn.~\ref{eqn:A'Cases} can equivalently be expressed as:
    \begin{align}
        \mathcal{A}'[\mathcal{S}^{\Gamma}_{\tau}]
        &=
        H[\tau] \cdot \big( 1 - \textit{Poly}(N_{\phi}) c^{N_{\phi}} \big)\\
        &+ \big( 1 - H[\tau] \big) \cdot \textit{Poly}(N_{\phi})c^{N_{\phi}} \nonumber\\
        \Rightarrow
        \mathcal{A}'[\mathcal{S}^{\Gamma}_{\tau}] 
        &=
        H[\tau] + (1 - 2H[\tau]) \cdot \textit{Poly}(N_{\phi}) c^{N_{\phi}}.
        \label{preTauEX}
    \end{align}
    Then, using Def.~\ref{def:canonParam} and that $N_{\phi}(x)$ is always positive:
    \begin{align}
            \Gamma(x) 
            &=
            \Q(x) \bigg  \vert \sqrt{N_{\phi}(x)} \bigg \vert
            \Rightarrow 
            [ \Gamma(x) ]^2
            =
            N_{\phi}(x).
            \label{nDefinTau}
    \end{align}
    Therefore, using Eqn.~\ref{nDefinTau} to replace $N_{\phi}$  with $\tau^2$ in Eqn.~\ref{preTauEX}:
    \begin{align}
        \mathcal{A}'[\mathcal{S}^{\Gamma}_{\tau}]
        &=
        H[\tau] + (1 - 2H[\tau]) \cdot \textit{Poly}(\tau^2) c^{\tau^2}.
    \end{align}
    So the acceptance fraction is bounded entirely in terms of $\tau$. It is now possible to prove that the canonical parameter meets the requirements of Def.~\ref{PhaseDef}, which shows a phase transition occurs. I address each requirement individually, in order:\\
    $\underline{\bold{Requirement \text{ } 1)}}$ As $\tau \rightarrow \infty$,
    \begin{align}
       \mathcal{A}'[\mathcal{S}^{\Gamma}_{\tau}] &= 1 + (1 - 2) \textit{Poly}(\tau^2)c^{\tau^2}\\
        &=
        1 - \textit{Poly}(\tau^2)c^{\tau^2}
        \longrightarrow
        \nonumber
        1\\
        \Rightarrow \mathcal{A}[\mathcal{S}^{\Gamma}_{\tau}] &\longrightarrow
        1, \text{ monotonically (in terms of a bound on it)}.
    \end{align}
    $\underline{\bold{Requirement \text{ } 2)}}$ As $\tau \rightarrow -\infty$,
    \begin{align}
        \mathcal{A}'[\mathcal{S}^{\Gamma}_{\tau}] &= 0 + (1 - 0) \cdot \textit{Poly}(\tau^2) c^{\tau^2}
        =
        \textit{Poly}(\tau^2) c^{\tau^2}
        \longrightarrow
        0\\
        \Rightarrow \mathcal{A}[\mathcal{S}^{\Gamma}_{\tau}] &\longrightarrow
        0, \text{ monotonically (in terms of a bound on it)}.
    \end{align}
    $\underline{\bold{Requirement \text{ } 3)}}$ This requirement is the only one that does not concern the acceptance fraction; as such, it does not rely on earlier parts of this manuscript (aside from notation and definitions). 
    
    Let $E \in \mathbb{R}^+$. I can express $\Gamma(x)$ being within $[-E, E]$ as:
    \begin{align}
        -E \leq \Gamma(x) \leq E
         \iff  N_{\phi}(x)\leq E^2.
    \end{align}
    I then need to ask how many inputs have $N_{\phi}(x)$ take a particular value. If I fix the value of $N_{\phi}(x)$ to $A  \in \mathbb{Z}^+$. As $N_{\phi}(x) = \vert \phi (x) \vert$ and $\phi (x)$ is a P-isomorphic encoding (so it is $1-1$), there are $\vert \Sigma \vert^A$ possible inputs (i.e. elements of $\Sigma^*$) that have $N_{\phi}(x) = A$.
    Hence, the number of inputs, $x \in \Sigma^*$, s.t. $-E \leq \Gamma(x) \leq E$ is:
    \begin{align}
    \label{eqn:summationFormula}
        \sum^{E^2}_{j = 0} \bigg( \vert \Sigma \vert^j \bigg)
        &=
        \dfrac{\vert \Sigma \vert^{E^2 + 1} - 1}{\vert \Sigma \vert - 1}.
    \end{align}
    In requirement 3 of Def.~\ref{PhaseDef}, what is actually considered is the number of instances with distance from the threshold value in a specified range. Therefore, to meet this requirement, I look at how many $x \in \Sigma^*$ have $\Gamma(x)$ in the ranges:
    \begin{align}
        E_1 \leq \Gamma(x) \leq E_2 \hspace{0.25cm} \textit{ or } \hspace{0.25cm}
        -E_2 \leq \Gamma(x) \leq - E_1,
    \end{align}
    where $E_1, E_2 \in \mathbb{R}^+$ and $E_2 > E_1$.
    The number of instances (i.e. $x \in \Sigma^*$) such that $\Gamma (x)$ is in these ranges is (using Eqn.~\ref{eqn:summationFormula}):
    \begin{align}
        \sum^{E_2^2}_{j = E_1^2} \bigg( \vert \Sigma \vert^j \bigg)
        =
        \sum^{E_2^2}_{j = 0} \bigg( \vert \Sigma \vert^j \bigg) - 
        \sum^{E_1^2 - 1}_{j = 0} \bigg( \vert \Sigma \vert^j \bigg)
        =
        \dfrac {\vert \Sigma \vert^{E_2^2 + 1} - \vert \Sigma \vert^{E_1^2}}{\vert \Sigma \vert - 1}.
    \end{align}
    I now fix the width of the ranges being considered (i.e. setting $E_2 - E_1$ to a constant value) to be $\delta \in \mathbb{R}^+$ by setting $E_2 = E_1 + \delta$:
    \begin{align}
        \sum^{E_2^2}_{j = E_1^2} \bigg( \vert \Sigma \vert^j \bigg)
        &=
        \vert \Sigma \vert^{E_1^2} \dfrac{\vert \Sigma \vert^{ 2 \delta E_1 + \delta^2 + 1} - 1}{\vert \Sigma \vert - 1}.
    \end{align}
    $\dfrac{\vert \Sigma \vert^{ 2 \delta E_1 + \delta^2 + 1} - 1}{\vert \Sigma \vert - 1}$ is strictly increasing as $E_1$ increases. Therefore, the number of elements of $\Sigma^*$ with corresponding values of $\Gamma(x)$ in $[E_1, E_1 + \delta]$, regardless of the value of $\delta$, is:
    \begin{align}
        o \bigg( \vert \Sigma \vert^{E_1^2} \bigg),
    \end{align}
    where $o \big( \cdot \big)$ denotes that anything it is equated to grows at least as quickly as its argument, in the asymptotic limit (defined more formally in Ref.~\cite{omnicronDef}).

    Therefore, as the value of $\Gamma(x)$ moves away from the threshold value, the number of instances found in $[\Gamma(x), \Gamma(x) + \delta]$ increases exponentially (as $\vert \Sigma \vert > 1$). This may alternatively be expressed as: approaching the threshold value results in the number of instances within a given width (i.e. size of the interval being considered, $\delta$) decreasing exponentially. This can be interpreted as very few inputs having a value of $\Gamma$ close (by any measure) to the threshold value, which provides the sparsity around the transition required in the less formal definition of phase transitions (discussed in Sec.~\ref{defAndPrelim}).
\end{proof}
The phase transitions I have shown to exist have been quite permissive.
For comparison with the standard form on which I focused my discussion, Fig.~\ref{fig:ExampleImageProvenSoFar} shows an example (with $\textit{Poly}(n)$ being a constant) of the bounds which I have proven to exist.
\begin{figure}[h!]
    \centering
\includegraphics[width=0.45\textwidth]{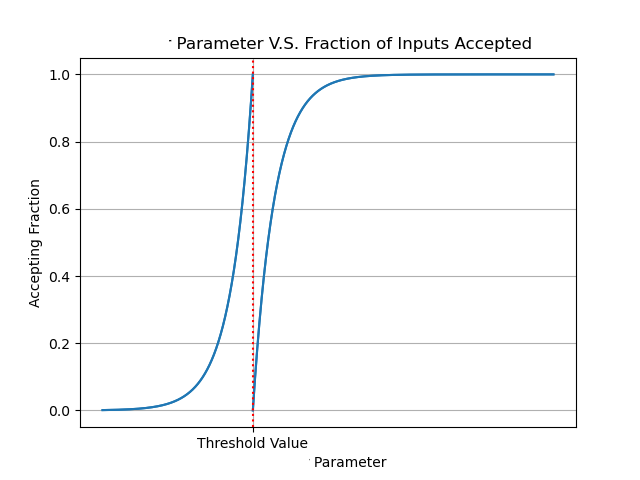}
    \caption{An example (where $\textit{Poly}(n)$ is a constant) of the limits that the acceptance fraction in any paddable decision problem over an even-sized alphabet have been shown to obey. The true acceptance fraction of any such decision problem, if superimposed on this figure would be constrained below the blue line when the parameter is below (in the above image, to the left of) the threshold value and constrained above the blue line when the parameter is above (in the above image, to the right of) the threshold value.}
    \label{fig:ExampleImageProvenSoFar}
\end{figure}

This ``permissive" form of phase transition is very widely applicable and many varying forms of phase transition would fit within this form. However, I note that transitions of the form in Fig.~\ref{fig:Example Image} would easily conform to the requirements of this ``permissive" form of phase transition. Visually, the phase transition from Fig.~\ref{fig:Example Image} can be seen to fall within these bounds in Fig.~\ref{fig:ExampleImageProvenSoFar}.

\section{Discussion}
The principal result of this paper has been to explain why phase transitions in decision problems are so common. In tying the presence of phase transitions to paddability, I have -- at least partially -- answered this question: paddability is extremely common in decision problems and is exhibited by almost all known practical problems. Likewise, there is no reason to suspect that not-anywhere-exponentially-unbalanced languages are rare. Hence the ubiquity of paddability explains the ubiquity of phase transitions.

This result has been achieved via a theoretical approach to establishing the presence of phase transitions in an extremely large set of decision problems -- the paddable not-anywhere-exponentially-unbalanced languages over even-sized alphabets. I note that phase transitions constructed by the methods herein are not necessarily the only phase transitions in a given decision problem. There are \emph{often} many restrictions that can be applied to a specific decision problem that makes it easy~\cite{jackson2024extensivelypbiimmunepromisebqpcompletelanguages}, or at least drastically easier, and these can be seen as phase transitions of a different form.
 
A potential improvement to the result of this paper is relatively obvious: the results of this paper are only applicable to languages over even-sized alphabets, so it would be useful and interesting to generalize the results herein to languages over any alphabet.

\section{Acknowledgements}
This work was supported, in part, 
by a EPSRC IAA grant (G.PXAD.0702.EXP) and the UKRI ExCALIBUR project QEVEC (EP/W00772X/2).

\bibliography{References}

\newpage
\onecolumngrid
\appendix 

\section{Proof of Lemma \ref{AccBoundLemma}}
\label{app:ProofOfLemma}
This appendix is devoted to the proof of Lemma~\ref{AccBoundLemma}. I first prove two preparatory lemmas: Lemma~\ref{equalSize} and Lemma~\ref{devidingLemma}.  Lemma~\ref{equalSize} is used in the proof of Lemma~\ref{devidingLemma}, which is then used in the proof of Lemma~\ref{AccBoundLemma}. Lemma~\ref{AccBoundLemma} also makes use of Lemma~\ref{lem:PAnalysis}, which appears immediately after it.
\subsection{Preparatory Lemmas: Lemma~\ref{equalSize} and Lemma~\ref{devidingLemma}}
\begin{lemma}
    \label{equalSize}
    For any even-sized alphabet, $\Sigma$, and $n \in \mathbb{N}$:
    \begin{align}
        \big \vert \big\{ x \in \Sigma^n \text{ } \vert \text{ } \omega \big( x \big) \textit{ is even} \big\} \big \vert
        &=
        \big \vert \big\{ x \in \Sigma^n \text{ } \vert \text{ } \omega \big( x \big) \textit{ is odd} \big\} \big \vert,
    \end{align}
    where $\omega: \Sigma^* \rightarrow \mathbb{N}$ is the weight function (i.e. the sum of each digit in the input string~\footnote{If the alphabet of the input string is not a set of sequential numbers, a bijection between the alphabet and $\mathbb{N}^{\leq \vert \Sigma \vert}$ can be used to enable the weight function}).
\end{lemma}
\begin{proof}
    Assume, WLOG, $\Sigma = \{ 1, 2, ..., \vert \Sigma \vert \}$, where $\vert \Sigma \vert$ is even.
    Letting $\big[ \cdot \big]_j: \Sigma^* \rightarrow \mathbb{N}$ denote the $j$th digit (note that I am using 1-based indexing) in its argument. I define, in Def.~\ref{def:omegaDef}, the weight function, $\omega: \Sigma^* \rightarrow \mathbb{N}$ that will feature prominently in this proof.
    \begin{definition}
        \label{def:omegaDef}
        Define $\omega: \Sigma^* \rightarrow \mathbb{N}$ by, $\forall x \in \Sigma^*$,
        \begin{align}
            \omega(x) &=  \sum_{j = 1}^{\vert x \vert} \bigg( \big[ x \big]_j \bigg).
        \end{align}
    \end{definition}
    Consider $x, y \in \Sigma^*$ and let $\circ: \Sigma^* \times \Sigma^* \rightarrow \Sigma^*$ denote string concatenation. Then,
    \begin{align}
    \label{eqn:digitsOfStringConcaternation}
    \big[ x \circ y \big]_j
        =
        \begin{cases}
            [x]_{j}, & \textit{if } j \leq \vert x \vert\\
            [y]_{j - \vert x \vert}, & \textit{if } j > \vert x \vert
        \end{cases}.
    \end{align}
    Applying the definition of $\omega: \Sigma^* \rightarrow \mathbb{N}$ to $x \circ y \in \Sigma^*$ and using Eqn.~\ref{eqn:digitsOfStringConcaternation}:
    \begin{align}
        \omega(x \circ y)
        &=
        \sum_{j = 1}^{\vert x \circ y \vert} \bigg( \big[ x \circ y \big]_j \bigg)
        =
        \sum_{j = 1}^{\vert x \vert} \bigg( \big[ x \big]_j \bigg)
        +
        \sum_{j = \vert x \vert + 1}^{\vert x \circ y \vert} \bigg( \big[ y \big]_{j - \vert x \vert} \bigg)
        =
        \sum_{j = 1}^{\vert x \vert} \bigg( \big[ x \circ y \big]_j \bigg)
        +
        \sum_{j = \vert x \vert + 1}^{\vert x \circ y \vert} \bigg( \big[ x \circ y \big]_j \bigg)
        =
        \omega(x)
        +
        \sum_{k =  1}^{\vert y \vert} \bigg( \big[ y \big]_{k} \bigg)
        =
        \label{omegaSumResult}
        \omega(x)
        +
        \omega(y).
    \end{align}
    Finally, define $\#_{E}(\cdot): \mathbb{N} \rightarrow  \mathbb{N}$ as the number of elements of $\Sigma^*$ of given length (which is the argument) such that the corresponding weight function is even, and similarly $\#_{O}(\cdot): \mathbb{N} \rightarrow  \mathbb{N}$ as the number of elements of $\Sigma^{*}$ of given length (its argument) such that the corresponding weight function is odd. Using Eqn.~\ref{omegaSumResult}: $\forall n, m \in \mathbb{N}$,
    \begin{align}
        \#_E(n + m)
        =
        \sum_{j \in \{E, O\}} \bigg( \#_j(n) \#_j(m) \bigg)
        \Rightarrow
        \#_E(n + 1)
        =
        \dfrac{\vert \Sigma \vert}{2} \sum_{j \in \{E, O\}} \bigg( \#_j(n) \bigg),
    \end{align}
    where I have used that $\#_E(1) = \#_O(1) = \dfrac{\vert \Sigma \vert}{2}$ (as $\vert \Sigma \vert$ is assumed to be even). A similar argument yields: $
        \#_O(n + 1)
        =
        \dfrac{\vert \Sigma \vert}{2} \sum_{j \in \{E, O\}} \bigg( \#_j(n) \bigg)$.
    Using the formulas for $\#_E(n + 1)$ and $\#_O(n + 1)$ (specifically, how they are identical), and that $\#_E(1) = \#_O(1)$ (as, for an even-sized alphabet like $\Sigma$, there are an equal number of odd and even weight strings of length one), by induction: $\forall q \in \mathbb{N}$,
    \begin{align}
        \#_E(q) = \#_O(q),
    \end{align} exactly as Lemma~\ref{equalSize} claims, albeit expressed more concisely.
\end{proof}

\begin{lemma}
    \label{devidingLemma}
    For any even $n \in \mathbb{N}$ and set $ \mathcal{q}_n =  \{ z \circ z \text{ } \vert \text{ } z \in \Sigma^{n/2} \}$, where $\Sigma$ is any even-sized alphabet, there exists a linear-time deterministic algorithm, $\mathfrak{Q}^{\prime}: \mathcal{q}_n \rightarrow \{ +1, -1 \}$, such that the pre-image of every element in $\{ +1, -1 \}$, under $\mathfrak{Q}'$ is of equal size. Additionally, the same $\mathfrak{Q}^{\prime}$ can be used for any even value of $n \in \mathbb{N}$. 
\end{lemma}
\begin{proof}
    Given an input $z \circ z $ such that $z \in \Sigma^{n/2}$, $z$ can be extracted in linear time.
    Then $\omega (z)$ can then be calculated in linear (in $\vert z \vert$) time. Via Lemma \ref{equalSize}, calculating $\omega (z)$ divides $\mathcal{q}_n$ into two equinumerous subsets: those where $\omega (z)$ is even and those where it is odd.
    Define $\mathfrak{Q}^{\prime}: \cup_{\text{even } n \in \mathbb{N}} \big( \mathcal{q}_n \big) \rightarrow \{ +1, -1 \}$ as computing -- on input $z \circ z$, where $z \in \Sigma^*$ -- $\omega(z)$ and returning $+1$ if the parity is odd and $-1$ otherwise. This equally divides the inputs (i.e. $\Sigma^{n/2}$) as required. It can also be seen that this same $\mathfrak{Q}^{\prime}: \cup_{\text{even } n \in \mathbb{N}} \big( \mathcal{q}_n \big) \rightarrow \{ +1, -1 \}$ functions as required (and in linear time) regardless of the value of $n \in \mathbb{N}$.
\end{proof}

\begin{Note}
    The $\mathfrak{Q}^{\prime}$ in Lemma~\ref{devidingLemma} is exactly the $\mathfrak{Q}^{\prime}$ in the definition of the discriminator (see Def.~\ref{discrimDef}). It only needs to be defined on $\cup_{n = 1}^{\infty} \big( \mathcal{q}_n \big)$ -- for my purposes -- as $\cup_{n = 1}^{\infty} \big( \mathcal{q}_n \big)$ is exactly the set of inputs that the discriminator -- as in Def.~\ref{discrimDef} -- passes to $\mathfrak{Q}^{\prime}$. This is due to $\cup_{n = 1}^{\infty} \big( \mathcal{q}_n \big)$ being exactly the set of inputs for which $\mathcal{P}$ -- as specified in Theorem~\ref{faragoRoughMain} -- returns $\perp$.
    I also note that we do not care about what $\mathfrak{Q}^{\prime}$ does outside of $\cup_{n = 1}^{\infty} \big( \mathcal{q}_n \big)$.
\end{Note}

\subsection{Presentation of Lemma~\ref{AccBoundLemma} and its Proof}
\begin{lemma}
\label{AccBoundLemma}
    For any paddable not-anywhere-exponentially-unbalanced decision problem over an even-sized alphabet, $\Sigma$, and its corresponding canonical parameter, $\Gamma$ (as defined in Eqn.~\ref{tauDef} and which exists for any such problem, due to Theorem~\ref{faragoRoughMain}), the accepting fraction of that decision problem (as defined in Def.~\ref{OGA}) over $\mathcal{S}_{\tau}^{\Gamma}$, $\mathcal{A}[\mathcal{S}_{\tau}^{\Gamma}]$, is bounded as:
    \begin{align}
       \textit{If } \tau > 0, \hspace{0.7em} \mathcal{A}[\mathcal{S}_{\tau}^{\Gamma}] &\geq 1 - \textit{Poly}(N_{\phi}) c^{N_{\phi}}.\\
         \textit{If } \tau < 0, \hspace{0.7em} \mathcal{A}[\mathcal{S}_{\tau}^{\Gamma}] &\leq \textit{Poly}(N_{\phi}) c^{N_{\phi}}.
    \end{align}
    Where $\textit{Poly}: \mathbb{R}^+ \rightarrow \mathbb{R}^+$ is a polynomial function as in Def.~\ref{PhiBalanced}, $N_{\phi}$ is as defined in Def.~\ref{pIsoDef} (but drops its argument with the understanding that it is the isomorphism length corresponding to $\Gamma$ taking the value $\tau \in \mathbb{R}$~\footnote{Where $\tau$ is as prescribed by the rest of the equation $N_{\phi}$ is in.}), and $\mathcal{S}^{\Gamma}_{\tau}$ retains its meaning from Def.~\ref{def:parameterSlice}.
\end{lemma}
\begin{proof}
I begin with a definition I will use throughout this proof, Def.~\ref{def:bracketsInApp}.
    \begin{definition}
    \label{def:bracketsInApp}
    For any map, $f$, defined on a set, $S$,~\footnote{Which need not be the domain of $f$ but must only be a subset of it.} define: $\underline{\big[ S \big]^{\textit{f}}_{\textit{val}}}
            =
            \big\{ x \in S \text{ } \vert \text{ } f(x) = \textit{val}  \big\}$,
        where $\textit{val}$ can be any element of $f$'s codomain.
Note that for any map, $f$, set, $S \subseteq \Sigma^*$, and value, $\textit{val}$; $\big[ S \big]^{\textit{f}}_{\textit{val}}$ can be expressed as the intersection of two subsets of $\Sigma^*$:
\begin{align}
    \label{eqn:altFormBrackets}
    \big[ S \big]^{\textit{f}}_{\textit{val}}
    &=
    S \cap \big\{ x \in \Sigma^* \text{ } \vert \text{ } f(x) = \textit{val} \big\}.
\end{align}
\end{definition}
    Using Def.~\ref{def:bracketsInApp}, I can, $\forall n \in \mathbb{N}$, split $\mathcal{B}^{\phi}_n$ (as defined in Def.~\ref{def:RoughP}) into $\big[ \mathcal{B}^{\phi}_n \big]^{\mathcal{Q}}_{+1}$ and $\big[ \mathcal{B}^{\phi}_n \big]^{\mathcal{Q}}_{-1}$ (as $\mathcal{Q}$ -- as defined in Def.~\ref{discrimDef} -- always outputs exactly one of only two options: $+1$ or $-1$) with the properties: $\forall n \in \mathbb{N}$,
    \begin{align}
        \big[ \mathcal{B}^{\phi}_n \big]^{\mathcal{Q}}_{+1} \cap \big[ \mathcal{B}^{\phi}_n \big]^{\mathcal{Q}}_{-1} = \emptyset
        \text{ and }
        \big[ \mathcal{B}^{\phi}_n \big]^{\mathcal{Q}}_{+1} \cup \big[ \mathcal{B}^{\phi}_n \big]^{\mathcal{Q}}_{-1} = \mathcal{B}^{\phi}_n.
    \end{align}
    Considering the cases by which $\mathcal{Q}$ is defined, in Def.~\ref{discrimDef}, and letting $\circ$ denote the composition operator on functions, I can re-express each of $\big[ \mathcal{B}^{\phi}_n \big]^{\mathcal{Q}}_{+1}$ and $\big[ \mathcal{B}^{\phi}_n \big]^{\mathcal{Q}}_{-1}$ as: $\forall n \in \mathbb{N}$, $\forall p \in \{+1, -1\}$,
    \begin{align}
        \label{firstSDecomp}
        \big[ \mathcal{B}^{\phi}_n \big]^{\mathcal{Q}}_{p}
        =
        \bigg( \big[ \mathcal{B}^{\phi}_n \big]^{\mathcal{P}}_{\perp} \cap \big[ \mathcal{B}^{\phi}_n \big]^{\mathfrak{Q}^{\prime} \circ \phi}_{p} \bigg)
        \cup
        \big[ \mathcal{B}^{\phi}_n \big]^{\mathcal{P}}_{\psi (p) }, \hspace{0.4cm}
    \text{ where }
        \psi(p)
        =
        \begin{cases}
            \textit{Acc}, &\textit{if } p = +1\\
            \textit{Rej}, &\textit{if } p = -1
        \end{cases}.
    \end{align}
    As the sets on either side of the logical union, in Eqn.~\ref{firstSDecomp}, are disjoint $\big($as $ \big[ \mathcal{B}^{\phi}_n \big]^{\mathcal{P}}_{\perp} \cap \big[ \mathcal{B}^{\phi}_n \big]^{\mathfrak{Q}^{\prime} \circ \phi}_{p} \subseteq \big[ \mathcal{B}^{\phi}_n \big]^{\mathcal{P}}_{\perp}$ and $\big[ \mathcal{B}^{\phi}_n \big]^{\mathcal{P}}_{\psi (p) } \cap \big[ \mathcal{B}^{\phi}_n \big]^{\mathcal{P}}_{\perp} = \emptyset\big)$:
    \begin{align}
    \label{Eqn:PEvalu}
        \big \vert \big[ \mathcal{B}^{\phi}_n \big]^{\mathcal{Q}}_{p} \big \vert 
        &=
        \big \vert \big[ \mathcal{B}^{\phi}_n \big]^{\mathcal{P}}_{\perp} \cap \big[ \mathcal{B}^{\phi}_n \big]^{\mathfrak{Q}^{\prime} \circ \phi}_{p} \big \vert
        +
        \big \vert \big[ \mathcal{B}^{\phi}_n \big]^{\mathcal{P}}_{\psi (p) } \big \vert.
    \end{align}
    An analysis of $\mathcal{P}$ -- in Lemma~\ref{lem:PAnalysis} -- allows me to evaluate: $\forall n \in \mathbb{N}$, $\forall p \in \{+1, -1\}$,
    \begin{align}
        \label{eqn:intermediateSizeEvalOne}
        \big \vert \big[ \mathcal{B}^{\phi}_n \big]^{\mathcal{P}}_{\psi (p) } \big \vert
        &=
        \dfrac{\vert \mathcal{B}^{\phi}_n \vert}{2}
        -
        \dfrac{(1 - p) \cdot \vert [\mathcal{B}^{\phi}_n ]_{\perp}^{\mathcal{P}} \vert}{2},
    \end{align}
    and as shown in Lemma~\ref{devidingLemma} (which depends on $\vert \Sigma \vert$ being even), $\mathfrak{Q}^{\prime} \circ \phi$, by construction, splits $\big[ \mathcal{B}^{\phi}_n \big]^{\mathcal{P}}_{\perp}$ evenly between $\big[ \mathcal{B}^{\phi}_n \big]^{\mathcal{Q}}_{+1}$ and $\big[ \mathcal{B}^{\phi}_n \big]^{\mathcal{Q}}_{-1}$, therefore: $\forall n \in \mathbb{N}$, $\forall p \in \{+1, -1\}$,
    \begin{align}
        \label{eqn:intermediateSizeEvalTwo}
        \big \vert \big[ \mathcal{B}^{\phi}_n \big]^{\mathcal{P}}_{\perp} \cap \big[ \mathcal{B}^{\phi}_n \big]^{\mathfrak{Q}^{\prime} \circ \phi}_{p} \big \vert
        &=
        \dfrac{\vert \big[ \mathcal{B}^{\phi}_n \big]^{\mathcal{P}}_{\perp} \vert}{2}.
    \end{align}
    I can then return to Eqn.~\ref{Eqn:PEvalu} and -- using Eqn.~\ref{eqn:intermediateSizeEvalOne} and Eqn.~\ref{eqn:intermediateSizeEvalTwo} -- re-express it as: $\forall n \in \mathbb{N}$, $\forall p \in \{+1, -1\}$,
    \begin{align}
        \label{eqn:NotGonnaEval}
        \big \vert \big[ \mathcal{B}^{\phi}_n \big]^{\mathcal{Q}}_{p} \big \vert 
        &=
        \dfrac{\vert \big[ \mathcal{B}^{\phi}_n \big]^{\mathcal{P}}_{\perp} \vert}{2}
        +
        \dfrac{\vert \mathcal{B}^{\phi}_n \vert}{2}
        -
        \dfrac{(1 - p) \cdot \vert [\mathcal{B}^{\phi}_n ]_{\perp}^{\mathcal{P}} \vert}{2}
        =
        \dfrac{\vert \mathcal{B}^{\phi}_n \vert}{2}
        +
        p \dfrac{ \vert [\mathcal{B}^{\phi}_n ]_{\perp}^{\mathcal{P}} \vert}{2}.
    \end{align}
    I do not need to evaluate the size of the sets ($\mathcal{B}^{\phi}_n$ and $[\mathcal{B}^{\phi}_n ]_{\perp}^{\mathcal{P}}$) in the right hand side of Eqn.~\ref{eqn:NotGonnaEval}. Instead I define $\mathcal{F}_p^{(\phi, n)}$ as the fraction of $\big[ \mathcal{B}^{\phi}_n \big]^Q_p$ where  $\mathcal{P}$ returns $\perp$ $ \bigg( \text{i.e. } \mathcal{F}_p^{(\phi, n)}
        =
        \dfrac{\big \vert \big[ \mathcal{B}^{\phi}_n \big]^{\mathcal{P}}_{\perp} \cap \big[ \mathcal{B}^{\phi}_n \big]^{\mathfrak{Q}^{\prime} \circ \phi}_{p} \big \vert}{\big \vert \big[ \mathcal{B}^{\phi}_n \big]^{\mathcal{Q}}_{p} \big \vert} \bigg)$. $\mathcal{F}_p^{(\phi, n)}$ is an upper bound on the fraction of $[ \mathcal{B}^{\phi}_n ]^Q_p$ that $Q$ decides incorrectly (if $+1$ and $-1$ are interpreted as Accept and Reject, respectively~\footnote{Where the ``correctness'' is defined relative to membership of the relevant language.}). Eqn.~\ref{eqn:intermediateSizeEvalTwo} and Eqn.~\ref{eqn:NotGonnaEval} then allow me to evaluate $\mathcal{F}_p^{(\phi, n)}$ as: $\forall n \in \mathbb{N}$, $\forall p \in \{+1, -1\}$,
    \begin{align}
    \label{Eqn:FFiguredOutABit}
        \mathcal{F}_p^{(\phi, n)}
        =
        \dfrac{\big \vert \big[ \mathcal{B}^{\phi}_n \big]^{\mathcal{P}}_{\perp} \cap \big[ \mathcal{B}^{\phi}_n \big]^{\mathfrak{Q}^{\prime} \circ \phi}_{p} \big \vert}{\big \vert \big[ \mathcal{B}^{\phi}_n \big]^{\mathcal{Q}}_{p} \big \vert}
        =
        \dfrac{ \vert [\mathcal{B}^{\phi}_n ]_{\perp}^{\mathcal{P}} \vert}{ 2 \big \vert \big[ \mathcal{B}^{\phi}_n \big]^{\mathcal{Q}}_{p} \big \vert}
        =
    \dfrac{ \vert [\mathcal{B}^{\phi}_n ]_{\perp}^{\mathcal{P}} \vert}{\big \vert \mathcal{B}^{\phi}_n \big \vert
        +
        p \big \vert [\mathcal{B}^{\phi}_n ]_{\perp}^{\mathcal{P}} \big \vert}.
    \end{align}
    In a brief aside -- that will allow me to upper bound $\mathcal{F}_p^{(\phi, n)}$ -- consider the derivative of $g: \mathbb{R} \rightarrow \mathbb{R}$, defined by $g(x) = \frac{x}{\beta + p x}$ (where $\beta \in \mathbb{R}^+$~\footnote{Note that I consider $\mathbb{R}^+$ to not include zero.}, $p \in \{ -1, +1 \}$, and assuming $x < \beta$), to show $g$ is strictly increasing on $\mathbb{R}$: $\forall \beta \in \mathbb{R}^+, \forall p \in \{ + 1, -1\}$,
    \begin{align}
        \dfrac{\partial g}{\partial x}
        =
        \dfrac{\beta}{(\beta + p x)^2}
        \Rightarrow \forall x \in \mathbb{R}, \dfrac{\partial g}{\partial x} &> 0
        \label{eqn:showIncreasing}
        \Rightarrow
        \forall x \in \mathbb{R}, \forall h \in \mathbb{R}^+, \dfrac{x + h}{\beta + p(x + h)} > \dfrac{x}{\beta + px}.
    \end{align}
Returning to Eqn.~\ref{Eqn:FFiguredOutABit}, I use Eqn.~\ref{eqn:showIncreasing} (which is valid as $\vert \mathcal{B}^{\phi}_n \vert \in \mathbb{R}^+$ and $\vert \mathcal{B}^{\phi}_n \vert > \vert [\mathcal{B}^{\phi}_n ]_{\perp}^{\mathcal{P}} \vert$) to upper-bound $\mathcal{F}_p^{(\phi, n)}$ by upper-bounding $\vert [\mathcal{B}^{\phi}_n ]_{\perp}^{\mathcal{P}} \vert$ using the main result of Ref.~\cite{farago2016roughly} (and Eqn.~\ref{eqn:RoughPDefiningEquation} in Def.~\ref{def:RoughP}): $\forall n \in \mathbb{N}$, $\forall p \in \{ -1, +1\}$,
    \begin{align}
        \label{eqn:FinalForF}
        \forall n, \big \vert \mathcal{B}^{\phi}_n \big \vert \in \mathbb{N},
        \dfrac{ c^{n} \big \vert \mathcal{B}^{\phi}_n \big \vert}{\big \vert \mathcal{B}^{\phi}_n \big \vert
        -
        c^{n} \vert \mathcal{B}^{\phi}_n \vert }
        \geq
        \dfrac{ c^{n} \big \vert \mathcal{B}^{\phi}_n \big \vert }{\big \vert \mathcal{B}^{\phi}_n \big \vert
        +
        c^{n} \big \vert \mathcal{B}^{\phi}_n \big \vert }
        \Rightarrow
         \forall n, \big \vert \mathcal{B}^{\phi}_n \big \vert \in \mathbb{N},
        \mathcal{F}_p^{(\phi, n)} 
        &\leq
        \dfrac{ c^{n} \vert \mathcal{B}^{\phi}_n \vert }{\big \vert \mathcal{B}^{\phi}_n \big \vert
        -
        c^{n} \big \vert \mathcal{B}^{\phi}_n \big \vert }
        =
        \dfrac{c^{n}}{1
        -
        c^{n}}
        \leq
        \dfrac{\sqrt{2}}{\sqrt{2} - 1}
        c^{n}
        \leq
        \dfrac{7}{2} c^{n},
    \end{align}
where the third inequality in Eqn.~\ref{eqn:FinalForF} follows from $c \leq 1 / \sqrt{2}$, as in Theorem~\ref{faragoRoughMain}.
$\mathcal{F}_p^{(\phi, n)}$ can be interpreted as the fraction, in the worst case, of the set of inputs -- in $[\mathcal{B}^{\phi}_n]^Q_p$ -- for which the discriminator, $\mathcal{Q}$, decides incorrectly (and hence $\mathcal{Q}$ outputs a value on the wrong side of zero, i.e. the returned value of the parameter indicates to Accept when the input isn't in the language or Reject when it is).

Due to Eqn.~\ref{nDefinTau}, if $\tau \in \mathbb{R}$ and $N_{\phi} = N_{\phi}(z)$ (where $z \in \Sigma^*$ is any input such that $\vert \Gamma(z) \vert = + \tau$), then:
\begin{align}
    \mathcal{S}_{+\tau}^{\Gamma} \cup \mathcal{S}_{-\tau}^{\Gamma} = \mathcal{S}_{ \vert \tau \vert}^{\vert \Gamma \vert} = \mathcal{S}^{N_{\phi}}_{N_{\phi}} = \mathcal{B}^{\phi}_{N_{\phi}} \text{ and } \mathcal{S}_{+\tau}^{\Gamma} \cap \mathcal{S}_{-\tau}^{\Gamma} = \emptyset. \hspace{1.5cm}
    \text{ This is as: } \mathcal{S}_{+\tau}^{\Gamma} = \big[ \mathcal{B}^{\phi}_{N_{\phi}} \big]^{\mathcal{Q}}_{+1} \text{ and } \mathcal{S}_{-\tau} ^{\Gamma} = \big[ \mathcal{B}^{\phi}_{N_{\phi}} \big]^{\mathcal{Q}}_{-1}.
\end{align}
 I now consider bounding the true acceptance fraction of $\mathcal{S}_{\tau}^{\Gamma}$, denoted $\mathcal{A}[\mathcal{S}_{\tau}^{\Gamma}]$. As $\vert \tau \vert$ is uniquely determined by $N_{\phi}$ (and \emph{vice versa}, as per Eqn.~\ref{def:canonParam}), $\vert \Gamma (x) \vert = + \tau$ whenever the value of $N_{\phi}$ is fixed. But presupposing the value of $\vert \Gamma (x) \vert = + \tau$, considerations can be split into two distinct cases, depending on the sign of $\tau$. These two cases are intended to approximate Accept and Reject (i.e. if a given element of $\Sigma^*$ is in the language being considered) but are imperfect. This follows from $\mathcal{Q}$ decides which of $\mathcal{S}^{\Gamma}_{+\tau}$ or $\mathcal{S}^{\Gamma}_{-\tau}$ and element of $\mathcal{B}^{\phi}_{N_{\phi}}$ is sorted into (as per Def.~\ref{def:canonParam}) and $\mathcal{Q}$ always agrees with $\mathcal{P}$ (a roughP algorithm and errorless heuristic) when it does not return $\perp$.\\
The immediately subsequent analysis of the two cases mentioned makes use of a new function, defined in Def.~\ref{def:BubbleT}.
\begin{definition}
\label{def:BubbleT}
    Let \underline{$\mathbb{T}: \Sigma^* \rightarrow \big \{ \text{Accept}, \text{Reject} \big \} $} be a function that correctly decides membership of the language being considered.
\end{definition}
The two cases that $\mathcal{B}^{\phi}_{N_{\phi}}$ is split into are based on whether $\tau$ is positive or negative and are:\\

\noindent \underline{\emph{Case One}: $\tau$ $>$ 0} \hspace{0.7cm}
In this case, $\mathcal{Q}$ deciding incorrectly means it has returned Accept (implicitly, i.e. it has returned $+1$) when it should have returned Reject (i.e. should have returned $-1$), hence the ideal acceptance fraction is one (and $\mathcal{F}_p^{(\phi, N_{\phi})}$ is the deviation from it, due to $\mathcal{Q}$ erring (due to $\mathcal{P}$ returning $\bot$), in the worst case).
Assuming the above worst case scenario, $\mathcal{A}[\mathcal{B}^{\phi}_{N_{\phi}}]$ will be one minus the fraction of $\mathcal{B}^{\phi}_{N_{\phi}}$ where $\mathcal{P}$ returns $\perp$. Therefore, if this worst case scenario is not assumed, $\mathcal{A}[\mathcal{S}^{\Gamma}_{\tau}]$ will be greater than or equal to this, i.e.:
\begin{align}
\label{eqn.bound1A}
    \mathcal{A}[\mathcal{S}^{\Gamma}_{\tau}]
    &=
    \dfrac{\big \vert \big[ \mathcal{S}^{\Gamma}_{\tau} \big]^{\mathbb{T}}_{\text{Accept}} \big \vert}{ \big \vert \mathcal{S}^{\Gamma}_{\tau} \big \vert }.
\end{align}
I then consider $\big[ \mathcal{S}^{\Gamma}_{\tau} \big]^{\mathbb{T}}_{\text{Accept}}$ and bound the number of elements it has.
\begin{align}
    \big \vert \big[ \mathcal{S}^{\Gamma}_{\tau} \big]^{\mathbb{T}}_{\text{Accept}} \big \vert
    \geq
    \big \vert \mathcal{S}^{\Gamma}_{\tau} \big \vert 
    &-
    \big \vert \big \{ x \in \mathcal{B}^{\phi}_{N_{\phi}} \text{ } \big \vert \text{ } x \text{ is included in } \mathcal{S}^{\Gamma}_{\tau} \text{ and is not in the language being considered} \big \} \big \vert.
\end{align}
As $\mathcal{P}$ is an errorless heuristic, and all elements in $\big \{ x \in \mathcal{B}^{\phi}_{N_{\phi}} \text{ } \big \vert \text{ } x \text{ is included in } \mathcal{S}^{\Gamma}_{\tau} \text{ and is not in the language being considered} \big \}$ must have been decided incorrectly by $\mathcal{Q}$, $x \in \mathcal{B}^{\phi}_{N_{\phi}}$ can only be misassigned  as it has been if $x \in \big[ \mathcal{B}^{\phi}_{N_{\phi}} \big]^{\mathcal{P}}_{\perp}$. Due to the definition of $\mathcal{F}_p^{(\phi, N_{\phi})}$, $ \big \vert \big[ \mathcal{B}^{\phi}_n \big]^{\mathcal{P}}_{\perp} \big \vert \leq 2 \mathcal{F}_p^{(\phi, N_{\phi})} \big \vert \mathcal{B}^{\phi}_{N_{\phi}} \big \vert$. This then implies that:     
\begin{align}
    &\big \vert \big \{ x \in \mathcal{B}^{\phi}_{N_{\phi}} \text{ } \big \vert \text{ } x \text{ is included in } \mathcal{S}^{\Gamma}_{\tau} \text{ and is not in the language being considered} \big \} \big \vert \leq  2 \mathcal{F}_p^{(\phi, N_{\phi})} \big \vert \mathcal{B}^{\phi}_{N_{\phi}} \big \vert.
\end{align}
Returning to Eqn.~\ref{eqn.bound1A}, the above analysis allows me to derive the bound:
\begin{align}
    \mathcal{A}[\mathcal{S}^{\Gamma}_{\tau}]
    &\geq
    \dfrac{\big \vert \mathcal{S}^{\Gamma}_{\tau} \big \vert  - 2 \mathcal{F}_p^{(\phi, N_{\phi})} \big \vert \mathcal{B}^{\phi}_{N_{\phi}} \big \vert}{ \big \vert \mathcal{S}^{\Gamma}_{\tau} \big \vert }
    =
    1 - 2 \dfrac{\mathcal{F}_p^{(\phi, N_{\phi})} \big \vert \mathcal{B}^{\phi}_{N_{\phi}} \big \vert}{ \big \vert \mathcal{S}^{\Gamma}_{\tau} \big \vert }
\end{align}
Using that $\mathcal{S}_{+\tau}^{\Gamma} \cup \mathcal{S}_{-\tau}^{\Gamma} = \big[ \mathcal{B}^{\phi}_{N_{\phi}}\big]^{\mathcal{Q}}_{+1} \cup  \big[ \mathcal{B}^{\phi}_{N_{\phi}}\big]^{\mathcal{Q}}_{-1} = \mathcal{B}^{\phi}_{N_{\phi}}$ and that the the language being considered is not-anywhere-exponentially-unbalanced (plus the interpretation of $\mathcal{S}_{+\tau}^{\Gamma}$ and $\mathcal{S}_{-\tau}^{\Gamma}$ as approximations of which elements of $\mathcal{B}^{\phi}_{N_{\phi}}$ are and are not, respectively, in the language being considered); $\forall q \in \{ +\tau, -\tau\}$, I can bound $\vert \mathcal{S}_{q}^{\Gamma} \vert$ by: $\vert \mathcal{S}_{q}^{\Gamma} \vert 
    \geq \big( \textit{Poly}(N_{\phi}) \big)^{-1} \vert \mathcal{B}^{\phi}_{N_{\phi}} \vert$, where $\textit{Poly}$ is the polynomial as in Def.~\ref{PhiBalanced} and -- as usual -- $N_{\phi}$ is inferred from the choice of $\tau$. Therefore,
\begin{align}
    \label{eqn:firstAcceptingBoundingBound}
    \mathcal{A}[\mathcal{S}_{\tau}^{\Gamma}] 
       &\geq
    1 - 7\dfrac{c^{N_{\phi}} \vert \mathcal{B}^{\phi}_{N_{\phi}} \vert}{\vert \mathcal{S}_{\tau}^{\Gamma} \vert} 
    \geq
    1 - 7 \dfrac{c^{N_{\phi}} \vert \mathcal{B}^{\phi}_{N_{\phi}} \vert}{\big( \textit{Poly}(N_{\phi}) \big)^{-1} \vert \mathcal{B}^{\phi}_{N_{\phi}} \vert}
    =
    1 - 7 \textit{Poly}(N_{\phi}) c^{N_{\phi}}.
\end{align}
\underline{\emph{Case Two: }$\tau$ $<$ 0} \hspace{0.7cm}
Now $\mathcal{Q}$ deciding incorrectly means it has returned Reject (again, implicitly: returning $-1$) when it should have returned Accept (i.e. $+1$), hence the ideal acceptance fraction is zero. Using a very similar argument to the case where $\tau$ $>$ 0, I start by considering $\big[ \mathcal{S}^{\Gamma}_{\tau} \big]^{\mathbb{T}}_{\text{Accept}}$ and again bound the number of elements it has. 

When $\tau < 0$, any element of $\big[ \mathcal{S}^{\Gamma}_{\tau} \big]^{\mathbb{T}}_{\text{Accept}}$ must be incorrectly decided by $\mathcal{Q}$. Again, this can only happen if that element is in $\big[ \mathcal{B}^{\phi}_{N_{\phi}} \big]^{\mathcal{P}}_{\perp}$. Therefore,
\begin{align}
    \big \vert \big[ \mathcal{S}^{\Gamma}_{\tau} \big]^{\mathbb{T}}_{\text{Accept}} \big \vert
    \leq
    \big \vert \big[ \mathcal{B}^{\phi}_n \big]^{\mathcal{P}}_{\perp} \big \vert \leq 2\mathcal{F}_p^{(\phi, N_{\phi})} \big \vert \mathcal{B}^{\phi}_{N_{\phi}} \big \vert.
\end{align}
Similarly to the first case above, this implies that:
\begin{align}
\label{eqn:secondAcceptingBoundingBound}
        \mathcal{A}[\mathcal{S}_{\tau}^{\Gamma}] =
        \dfrac{\big \vert \big[ \mathcal{S}^{\Gamma}_{\tau} \big]^{\mathbb{T}}_{\text{Accept}} \big \vert}{ \big \vert \mathcal{S}^{\Gamma}_{\tau} \big \vert }
        &\leq
        \dfrac{\mathcal{F}_p^{(\phi, N_{\phi})} \big \vert \mathcal{B}^{\phi}_{N_{\phi}} \big \vert}{ \big \vert \mathcal{S}^{\Gamma}_{\tau} \big \vert }
        \leq
         7 \dfrac{c^{N_{\phi}} \vert \mathcal{B}^{\phi}_{N_{\phi}} \vert}{\vert \mathcal{S}_{\tau}^{\Gamma} \vert}
        \leq 
         7 \dfrac{c^{N_{\phi}} \vert \mathcal{B}^{\phi}_{N_{\phi}} \vert}{\big( \textit{Poly}(N_{\phi}) \big)^{-1} \vert \mathcal{B}^{\phi}_{N_{\phi}} \vert}
         =
         7 \textit{Poly}(N_{\phi}) c^{N_{\phi}}.
    \end{align}
    To complete the proof, the factor of $7$ (in Eqn.~\ref{eqn:secondAcceptingBoundingBound}) is absorbed into the polynomial, $\textit{Poly}$, and the two cases in Eqn.~\ref{eqn:firstAcceptingBoundingBound} and Eqn.~\ref{eqn:secondAcceptingBoundingBound} are combined into:
    \begin{align}
       \textit{If } \tau > 0, \hspace{0.7em} \mathcal{A}[\mathcal{S}_{\tau}^{\Gamma}] &\geq 1 - \textit{Poly}(N_{\phi}) c^{N_{\phi}}.\\
         \textit{If } \tau < 0, \hspace{0.7em} \mathcal{A}[\mathcal{S}_{\tau}^{\Gamma}] &\leq \textit{Poly}(N_{\phi}) c^{N_{\phi}}.
    \end{align}
    Where $\textit{Poly}: \mathbb{R}^+ \rightarrow \mathbb{R}^+$ is a polynomial function as in Def.~\ref{PhiBalanced}, $N_{\phi}$ is as defined in Def.~\ref{pIsoDef}, and $\mathcal{S}^{\Gamma}_{\tau}$ retains its meaning from Def.~\ref{def:parameterSlice}.
\end{proof}
\begin{Note}
    \label{note:QvaluesInAppendix}
    Due to the construction of $\Gamma$, $\Gamma(x) > 0 \iff \mathcal{Q} = +1$, and $\Gamma(x) < 0 \iff \mathcal{Q} = -1$, as can be seen from Def.~\ref{def:canonParam}.
\end{Note}
\subsection{ Lemma~\ref{lem:PAnalysis} and its Proof}
\begin{lemma}
    \label{lem:PAnalysis}
    Given the notation and assumptions of Lemma~\ref{AccBoundLemma} (and its proof),
        $\big \vert \big[ \mathcal{B}^{\phi}_n \big]^{\mathcal{P}}_{\psi (p) } \big \vert
        =
        \dfrac{\vert \mathcal{B}^{\phi}_n \vert}{2}
        -
        \dfrac{(1 - p) \cdot \vert [\mathcal{B}^{\phi}_n ]_{\perp}^{\mathcal{P}} \vert}{2}$.
\end{lemma}
\begin{proof}
    Consider the functioning of $\mathcal{P}: \Sigma^* \rightarrow \{ \text{Accept}, \text{Reject}, \perp \}$ from Eqn.~4 in Ref.~\cite{farago2016roughly}:
    \begin{align}
        \mathcal{P}(x)
        &= 
        \begin{cases}
            \text{Accept} &\text{ if } \omega(\phi(x)) \text{ is odd} \\
            \text{Reject} &\text{ if } \omega(\phi(x)) \text{ is even and } \phi(x) \text{ is asymmetric}\\
            \perp &\text{ if } \phi(x) \text{ is symmetric}
        \end{cases},
    \end{align}
    where $\omega: \Sigma^* \rightarrow \mathbb{N}_0$ is the weight function and $\phi: \Sigma^* \rightarrow \Sigma^*$ is a P-isomorphism that always exists for paddable languages, as in Theorem~\ref{faragoRoughMain}. Considering just the $ x \in \Sigma^*$ where $x \in \mathcal{B}^{\phi}_n$, as $\phi$ is a P-isomorphism so is $\phi^{-1}$, therefore $\mathcal{B}^{\phi}_n$ may be re-expressed as:
    \begin{align}
        \label{eqn:B_nEquating}
        \mathcal{B}^{\phi}_n
        &=
        \big \{ x \in \Sigma^* \text{ } \vert \text{ }  N_{\phi}(x) = n \big \}
        =
        \big \{ x \in \Sigma^*  \text{ } \vert \text{ }  \phi(x) \in \Sigma^n \big \}
        \Rightarrow 
        \text{ Letting } \phi(x) = z \text{, }
        \mathcal{B}^{\phi}_n
        =
        \big \{ \phi^{-1}(z) \text{ }  \text{ } \vert\text{ } \phi(\phi^{-1}(z)) \in \Sigma^n \big \}
        =
        \big \{ \phi^{-1} (z) \text{ }  \text{ } \vert \text{ }  z \in  \Sigma^n \big \}.
    \end{align}
    Then consider the subset of $\mathcal{B}^{\phi}_n$ for which $\mathcal{P}$ returns Accept, using Eqn.~\ref{eqn:B_nEquating}:
    \begin{align}
        \big \vert \big[ \mathcal{B}^{\phi}_n \big]^{\mathcal{P}}_{\psi (1)} \big \vert
        &=
        \big \vert \big[ \mathcal{B}^{\phi}_n \big]^{\mathcal{P}}_{\text{Acc}} \big \vert
        =
        \big \vert \big \{x \in \mathcal{B}^{\phi}_n \text{ } \vert \text{ }\omega(\phi(x)) \text{ is odd} \big \} \big \vert
        =
        \big \vert \big \{ y \in \Sigma^n \text{ } \vert \text{ } \omega(\phi(\phi^{-1}(y))) \text{ is odd} \big \} \big \vert
         =
         \big \vert \big \{ y \in \Sigma^n \text{ } \vert \text{ } \omega(y) \text{ is odd} \big \} \big \vert
         =
         \dfrac{\big \vert \Sigma^n \big \vert}{2}
         =
         \label{p=1Case}
         \dfrac{\big \vert \mathcal{B}^{\phi}_n \big \vert}{2},
    \end{align}
    where $\Sigma^n = \big \{ x  \in \Sigma^* \text{ } \vert \text{ } \vert x \vert = n \big\}$ (i.e. the set of all strings in $\Sigma^*$ of length $n \in \mathbb{N}$) and $\text{Acc}$ is used as a shorthand for $\text{Accept}$.
    The final equality in Eqn.~\ref{p=1Case} follows as $\phi$ is bijective.
    Similarly, consider the set of elements in $\mathcal{B}^{\phi}_n$ for which $\mathcal{P}$ returns Reject:
    \begin{align}
    \label{p=-1Case}
        \big \vert \big[ \mathcal{B}^{\phi}_n \big]^{\mathcal{P}}_{\psi (-1) } \big \vert
        &=
        \big \vert \big[ \mathcal{B}^{\phi}_n \big]^{\mathcal{P}}_{\text{Rej}} \big \vert
        =
        \big \vert \big \{ x \in \mathcal{B}^{\phi}_n \text{ } \vert \text{ } \omega(\phi(x)) \text{ is even and } \phi(x) \text{ is asymmetric} \big \} \big \vert\\
        &=
        \big \vert \big \{y \in \Sigma^n \text{ } \vert \text{ } \omega(\phi(\phi^{-1}(y))) \text{ is even and } \phi(\phi^{-1}(y)) \text{ is asymmetric} \big \} \big \vert
         =
        \big \vert \big \{y \in \Sigma^n \text{ } \vert \text{ } \omega(y) \text{ is even and } y \text{ is asymmetric} \big \} \big \vert \nonumber\\
        &=
       \big \vert  \big \{y \in \Sigma^n \text{ } \vert \text{ } \omega(y) \text{ is even}  \big \}
        \cap
        \big \{y \in \Sigma^n \text{ } \vert \text{ } y \text{ is asymmetric}  \big \} \big \vert
        =
        \big \vert \big \{ y \in \Sigma^n \text{ } \vert \text{ } \omega(y) \text{ is even } \big \} \bigg \backslash
        \big \{ y \in \Sigma^n \text{ } \vert \text{ } y \text{ is symmetric} \big \} \big \vert. \nonumber
    \end{align}
    Note that $\text{Rej}$ is used as a shorthand for $\text{Reject}$ and the backslash, i.e. $A \backslash B$ (as in Eqn.~\ref{p=-1Case}), above is used to denote the set complement of the set $B$ within the set $A$.
    
    As $\omega(y)$ is even for any symmetric $y \in \Sigma^*$ and the number of symmetric strings of length $n \in \mathbb{N}$ is exactly the size of $ [\mathcal{B}^{\phi}_n ]_{\perp}^{\mathcal{P}}$, Eqn.\ref{p=-1Case} becomes:
    \begin{align}
    \label{eqn:penultimeateEquationLemma5}
        \big \vert \big[ \mathcal{B}^{\phi}_n \big]^{\mathcal{P}}_{\psi (-1) } \big \vert
         =
         \dfrac{ \big \vert \Sigma^n \big \vert}{2} - \big \vert \big \{y \in \Sigma^n \text{ } \vert \text{ }  y \text{ is symmetric} \big \} \big \vert
         =
         \dfrac{\big \vert \mathcal{B}^{\phi}_n \big \vert}{2}
         - \big \vert [\mathcal{B}^{\phi}_n ]_{\perp}^{\mathcal{P}} \big \vert. 
    \end{align}
I combine the two cases (with differing $p \in \{ -1, +1 \}$) in Eqn.~\ref{p=1Case} and Eqn.~\ref{eqn:penultimeateEquationLemma5} into a single equation for $\big \vert \big[ \mathcal{B}^{\phi}_n \big]^{\mathcal{P}}_{\psi (p) } \big \vert$: $\forall n \in \mathbb{N}$, $\forall p \in \{ + 1, -1 \}$, 
    \begin{align}
        \big \vert \big[ \mathcal{B}^{\phi}_n \big]^{\mathcal{P}}_{\psi (p) } \big \vert
        &=
        \dfrac{\big \vert \mathcal{B}^{\phi}_n \big \vert}{2}
        -
        \dfrac{(1 - p) \cdot \big \vert [\mathcal{B}^{\phi}_n ]_{\perp}^{\mathcal{P}} \big \vert}{2}.
    \end{align}
\end{proof}

\end{document}